\RequirePackage{amsmath}
\documentclass[a4paper, 10pt]{amsart}
\usepackage{amssymb, url, color, pb-diagram, graphicx, amscd, pb-diagram,mathrsfs}
\usepackage[nodayofweek]{datetime}
\usepackage{enumerate,stmaryrd} 
\usepackage{euscript}
\usepackage{verbatim}
\usepackage{graphicx}
\usepackage{times}
\usepackage{setspace}
\usepackage[space, compress, sort]{cite}
\usepackage{empheq}

\numberwithin{equation}{section}

\theoremstyle{plain}

\newtheorem{thm}{Theorem}[section]
\newtheorem{theorem}[thm]{Theorem}

\newtheorem{corollary}[thm]{Corollary}

\newtheorem{lemma}[thm]{Lemma}
\newtheorem{prop}[thm]{Proposition}
\newtheorem{proposition}[thm]{Proposition}

\theoremstyle{remark}

\newtheorem{remark}[thm]{Remark}

\theoremstyle{definition}

\newcounter{mnotecount}[section]

\renewcommand{\phi}{\varphi}

\newcommand{\bR}{\mathbb{R}}

\newcommand{\bS}{\mathbb{S}}
\newcommand{\bZ}{\mathbb{Z}}

\newcommand{\cA}{\mathcal{A}}

\newcommand{\cM}{\mathcal{M}}

\newcommand{\bg}{\mathbf{g}}
\newcommand{\bT}{\mathbf{T}}

\newcommand{\ghat}{{\widehat{g}}}
\newcommand{\Khat}{\widehat{K}}

\newcommand{\gbar}{\overline{g}}

\renewcommand{\hbar}{\overline{h}}

\newcommand{\gtil}{\widetilde{g}}

\def\into{\hookrightarrow}

\newcommand{\Wtil}{\widetilde{W}}

\newcommand{\nablabar}{\overline{\nabla}}

\newcommand{\definedas}{\mathrel{\raise.095ex\hbox{\rm :}\mkern-5.2mu=}}

\let\<\langle 
\let\>\rangle

\DeclareMathOperator{\tr}{tr}
\DeclareMathOperator{\divg}{div}

\newcommand{\lie}{\mathcal{L}}

\newcommand{\nablatil}{\widetilde{\nabla}}

\newcommand{\pdiff} [2]{\frac{\partial #1}{\partial #2}}

\newcommand{\ric}{\mathrm{Ric}}

\newcommand{\rictil}{\widetilde{\ric}}
\newcommand{\ricdd}[2]{\ric_{#1 #2}}

\newcommand{\ricbar}{\overline{\ric}}

\newcommand{\bric}{\mathbf{Ric}}

\newcommand{\scal}{\mathrm{Scal}}

\newcommand{\scalbar}{\overline{\scal}}

\newcommand{\bscal}{\mathbf{Scal}}

\newcommand{\hessbardd}[2]{\overline{\nabla}^2_{#1, #2}}
\newcommand{\hesstildd}[2]{\widetilde{\nabla}^2_{#1, #2}}

\begin{document}


\title[Limit equation criterion with a 1-parameter symmetry]
{Limit equation for vacuum Einstein constraints with a translational
Killing vector field in the compact hyperbolic case}

\author{Romain Gicquaud}
\address{Romain Gicquaud, Laboratoire de Math\'ematiques et de Physique Th\'eorique,
UFR Sciences et Technologie, Universit\'e Fran\c cois Rabelais, Parc de Grandmont,
37300 Tours, France} \email{romain.gicquaud@lmpt.univ-tours.fr}

\author{C\'ecile Huneau}
\address{C\'ecile Huneau, \'Ecole Normale Sup\'erieure,
D\'epartement de Math\'ematiques et Applications
45, rue d'Ulm 75005 Paris, France}
\email{cecile.huneau@ens.fr}

\begin{abstract}
We construct solutions to the constraint equations in general relativity
using the limit equation criterion introduced in \cite{DahlGicquaudHumbert}.
We focus on solutions over compact 3-manifolds admitting a $\bS^1$-symmetry
group. When the quotient manifold has genus greater than 2, we obtain strong
far from CMC results.
\end{abstract}

%
%
%

\date{\today}


\maketitle

\tableofcontents

\section{Introduction}

General relativity describes the universe as a $(3+1)$-dimensional
manifold $\cM$ endowed with a Lorentzian metric $\bg$. The Einstein
equations describe how non-gravitational fields influence the curvature
of $\bg$:
\[
\bric_{\mu\nu} - \frac{\bscal}{2} \bg_{\mu\nu} = 8 \pi \bT_{\mu\nu},
\]
where $\bric$ and $\bscal$ are respectively the Ricci tensor and the
scalar curvature of the metric $\bg$ and $\bT_{\mu\nu}$ is the sum of
the energy-momentum tensors of all the non-gravitational fields.

Einstein equations can be formulated as a Cauchy problem with initial
data given by a set $(M, \ghat, \Khat)$, where $M$ is a 3-dimensional
manifold, $\ghat$ is a Riemannian metric on $M$ and $\Khat$ is a symmetric
$2$-tensor on $M$. $\ghat$ and $\Khat$ correspond to the first and second
fundamental forms of $M$ seen as an embedded space-like hypersurface
in the universe $(\cM, \bg)$ solving the Einstein equations.

It turns out that the Einstein equations imply compatibility conditions
on $\ghat$ and $\Khat$ known as the constraint equations:

\begin{subequations}\label{eqConstraintEquations}
\begin{empheq}[left=\empheqlbrace]{align}
\label{eqHamiltonian}
\scal_\ghat + (\tr_\ghat \Khat)^2 - |\Khat|_\ghat^2 & =2 \rho,\\
\label{eqMomentum}
\divg_\ghat \Khat - d (\tr_\ghat \Khat ) & = j,
\end{empheq}
\end{subequations}
where, denoting by $N$ the unit future-pointing normal to $M$ in $\cM$,
one has
\[
\rho   = 8 \pi \bT_{\mu\nu} N^\mu N^\nu, \quad
j_i    = 8 \pi \bT_{i \mu} N^\mu.
\]

We assume here that $\mu$ and $\nu$ go from $0$ to $3$ and denote spacetime
coordinates while Latin indices go from $1$ to $3$ and correspond to
coordinates on $M$.

In this article, to keep things simple, we will consider no field but the
gravitational one (vacuum case). As a consequence, we impose $\bT \equiv 0$.
We will also assume that the spacetime possesses a $\bS^1$-symmetry
generated by a spacelike Killing vector field. This allows for a reduction
of the $(3+1)$-dimensional study of the Einstein equations to a
$(2+1)$-dimensional problem. This symmetry assumption has been introduced
and studied by Y. Choquet-Bruhat and V. Moncrief in
\cite{ChoquetBruhatMoncrief} (see also \cite{ChoquetBruhat}) in the case
of a spacetime of the form $\Sigma \times \bS^1 \times \bR$, where $\Sigma$
is a compact 2-dimensional manifold of genus $G \geq 2$, $\bS^1$ corresponds
to the orbit of the $\bS^1$-action and $\bR$ is the time axis. They proved
the existence of global solutions corresponding to perturbations of a
particular expanding spacetime. In \cite{ChoquetBruhatMoncrief}, they
use solutions of the constraint equations with constant mean curvature
(CMC, i.e. constant $\tr_{\ghat} \Khat$) on the spacelike hypersurface
$\Sigma \times \bS^1 \times \{0\}$ as initial data. The construction of
such solutions is fairly direct. In this article we shall generalize
their construction to more general initial data allowing for non-constant
mean curvature.

The method which is generally used to construct initial data for the
Einstein equations is the conformal method which consists in decomposing
the metric $\ghat$ and the second fundamental form $\Khat$ into given
data and unknowns that have to be adjusted so that $\ghat$ and $\Khat$
solve the constraint equations, see Section \ref{secPrelim}. The equations
for the unknowns, namely a positive function playing the role of a conformal
factor and a 1-form, are usually called the conformal constraint equations.
Extended discussion of the conformal method can be found in a series of
very nice articles by D. Maxwell \cite{MaxwellConformalParameterization,
MaxwellConformalMethod,MaxwellKasner,MaxwellInitialData}.

These equations have been extensively studied in the case of constant
mean curvature (CMC) since the system greatly simplifies in this case.
We refer the reader to the excellent review article \cite{BartnikIsenberg}
for an overview of known results in this particular case. The non-CMC
case remained open for a couple of decades. Only the case of nearly
constant mean curvature was studied. Two major breakthroughs were obtained
in \cite{HNT1}, \cite{MaxwellNonCMC} and \cite{DahlGicquaudHumbert}
concerning the far from CMC case. A comparison of these methods is given
in \cite{GicquaudNgo}.

In this article, we follow the method described in \cite{DahlGicquaudHumbert}.
Namely, we give the following criterion: if a certain limit equation admits
no non-zero solution, the conformal constraint equations admit at least one
solution. The other method \cite{HNT1, MaxwellNonCMC} would require that
$\Sigma$ is $\bS^2$ so that it carries a metric with positive scalar
curvature and has no conformal Killing vector field, which is impossible.

This approach has been generalized to the asymptotically hyperbolic case
in \cite{GicquaudSakovich} and to the asymptotically cylindrical case in
\cite{DiltsLeach}. The asymptotically Euclidean case
\cite{DiltsGicquaudIsenberg} and the case of compact manifolds with
boundary \cite{GicquaudHumbertNgo} are currently work in progress since
new ideas have to be found to get the criterion.

The outline of the paper is as follows. In Section \ref{secPrelim}, we
show how the Einstein equations reduce to a $(2+1)$-dimensional problem
in the case of a $\bS^1$-symmetry and exhibit the analog of the conformal
constraint equations in this case. We also state Theorem \ref{thmMain}
which is the main result of this article and Corollary \ref{corNearCMC}
which gives an example of application of Theorem \ref{thmMain}. Section
\ref{secProof} is devoted to the proof of Theorem \ref{thmMain}.
Finally, Section \ref{secEx} contains the proof of Corollary
\ref{corNearCMC}.

\section{Preliminaries}\label{secPrelim}

\subsection{Reduction of the Einstein equations}

Before discussing the constraint equations, we briefly recall the form of
the Einstein equations in the presence of a spacelike translational
Killing vector field. We follow here the exposition in \cite[Section XVI.3]{ChoquetBruhat}.

We recall that we want to write the Einstein equations on the manifold
$\cM = \Sigma \times \bS^1 \times \bR$, where $\Sigma$ is a Riemannian surface
and $\bR$ denotes the time direction, for some metric $\bg$ which is
invariant under translation along the
$\bS^1$-direction. We let $x^3$ denote the coordinate along the $\bS^1$-
direction (seen as $\bR/\bZ$), choose local coordinates $x^1$, $x^2$ on
$\Sigma$ and denote by $x^0$ the time coordinate.

A metric $\bg$ on $\cM$ admitting $\partial_3$ as a Killing vector field
has the form

\[
\bg = \gtil + e^{2\gamma} \left(dx^3 + A\right)^2,
\]
where $\gtil$ is a Lorentzian metric on $\Sigma \times \bR$, $A$ is a
1-form on $\Sigma \times \bR$ and $\gamma$ is a function on $\Sigma \times \bR$.
Since $\partial_3$ is a Killing vector field, $\gtil$, $A$ and $\gamma$ do
not depend on $x^3$. We set $F = dA$ the field strength of $A$. The Ricci
tensor $\bric$ of $\bg$ can be computed in terms of $\gtil$, $A$ and
$\gamma$. In the basis $(dx^0, dx^1, dx^2, dx^3 + A)$, the vacuum Einstein
equations ($\bric = 0$) become 

\begin{subequations}
\begin{empheq}[left=\empheqlbrace]{align}
\label{eqTangentialRic}
0 & = \bric_{\alpha\beta} = \rictil_{\alpha\beta}
      - \frac{1}{2} e^{2 \gamma} F_{\alpha}^{\phantom{\alpha}\lambda} F_{\beta\lambda} - \hesstildd{\alpha}{\beta} \gamma - \nabla_\alpha \gamma \nabla_\beta \gamma,\\
\label{eqMixedRic}
0 & = \bric_{\alpha3} = \frac{1}{2} e^{-\gamma} \nablatil_\beta \left(e^{3\gamma} F_\alpha^{\phantom{\alpha}\beta}\right),\\
\label{eqNormalRic}
0 & = \bric_{33} = -e^{-2\gamma} \left(-\frac{1}{4} e^{2\gamma} F_{\alpha\beta} F^{\alpha\beta}
      + \gtil^{\alpha\beta} \nabla_\alpha \gamma \nabla_\beta \gamma + \gtil^{\alpha\beta} \hesstildd{\alpha}{\beta} \gamma\right),
\end{empheq}
\end{subequations}
where the indices $\alpha$, $\beta$ and $\lambda$ go from $0$ to $2$, and are raised with respect to the metric $\gtil$.
The equation \eqref{eqMixedRic} is equivalent to $d(*e^{3\gamma} F) = 0$.
So we are going to assume that $*e^{3\gamma} F$ is an exact 1-form.
Therefore, there exists a potential $\omega: \Sigma \times \bR \to \bR$
such that $e^{3\gamma} F = d\omega$.

Defining $\gbar = e^{2\gamma} \gtil$, we obtain the following system for $\gbar$,
$\gamma$ and $\omega$: 

\begin{subequations}\label{eqSystemFinal}
\begin{empheq}[left=\empheqlbrace]{align}
\Box_{\gbar} \omega - 4 \nablabar^\alpha \gamma \nablabar_\alpha \omega & = 0,\\
\Box_{\gbar} \gamma - \frac{1}{2} e^{-4\gamma} \nablabar^\alpha \omega \nablabar_\alpha \omega & = 0,\\
\ricbar_{\alpha\beta} - 2 \nablabar_\alpha \gamma \nablabar_\beta \gamma - \frac{1}{2} e^{-4 \gamma} \nablabar_\alpha \omega \nablabar_\beta \omega & = 0,
\end{empheq}
\end{subequations}

where $\Box_{\gbar} = \gbar^{\alpha\beta} \hessbardd{\alpha}{\beta}$ is the d'Alembertian
associated to the metric $\gbar$, $\ricbar$ is its Ricci tensor and the indices are raised with respect to $\gbar$. We introduce the following notation

\[
 u \definedas (\gamma, \omega),
\]
together with the scalar product

\[
 \partial_\alpha u \cdot \partial_\beta u \definedas 2 \partial_\alpha \gamma \partial_\beta \gamma + \frac{1}{2} e^{-4\gamma} \partial_\alpha \omega \partial_\beta \omega.
\]

We are going to consider the Cauchy problem for the system \eqref{eqSystemFinal}.
As for the general Einstein equations, the initial data for this system
have to satisfy some constraint equations.

\subsection{The constraint equations}

We write the metric $\gbar$ under the following form:

\[\gbar = -N^2 dt^2 + g_{ij}\left(dx^i + \beta^i dt\right) \left(dx^i + \beta^j dt\right)\]

The coefficient $N$ is called the lapse, while the vector $\beta$ is called
the shift. $g$ is the Riemannian metric induced by $\gbar$ on the slices of constant
$t$. We consider the initial data for the spacelike surface $\Sigma$ which
is the constant $t=0$ hypersurface of $\Sigma \times \bR$. We also use
the notation \[\partial_t = \partial_0 - \lie_\beta,\] where $\lie_\beta$
is the Lie derivative associated to the vector field $\beta$. With this
notation, the second fundamental form of $\Sigma \subset \Sigma \times \bR$
reads \[K_{ij} = - \frac{1}{2N} \partial_t g_{ij}.\]
We denote by $\tau$ the mean curvature of $\Sigma$:

\[\tau \definedas g^{ij} K_{ij}.\]

The constraint equations are obtained by taking the $\partial_t-\partial_t$
and the $\partial_t-\partial_i$ components of the Einstein equations:

\begin{subequations}\label{eqConstraintEquations2}
\begin{empheq}[left=\empheqlbrace]{align}
\label{eqMomentum2}
\ricbar_{t i} - \frac{\scalbar}{2} \gbar_{t i} = N \left(\partial_i \tau - D^i K_{ij}\right) & = \partial_t u \cdot \partial_i u,\\
\label{eqHamiltonian2}
\ricbar_{t t} - \frac{\scalbar}{2} \gbar_{t t} = \frac{N^2}{2}\left(\scal - \left|K\right|^2 + \tau^2\right) & = \partial_t u \cdot \partial_t u + \frac{N^2}{2} \gbar^{\alpha\beta} \partial_\alpha u \cdot \partial_\beta u,\\
\end{empheq}
\end{subequations}

where $\scal$ is the scalar curvature of the metric $g$ and $D$ is its
Levi-Civita connection. Equation \eqref{eqMomentum2} is called the
\emph{momentum constraint} while Equation \eqref{eqHamiltonian2} is known
as the \emph{Hamiltonian constraint}.

\subsection{The conformal method}

In order to construct solutions to the system \eqref{eqConstraintEquations2},
we are going to use the well-known conformal method which we explain now.

Given a Riemann surface $\Sigma$ of genus $G \geq 2$, we let $g_0$ be a
metric on $\Sigma$ with constant scalar curvature $\scal_0 \equiv -1$ and
look for a metric $g$ in the conformal class of $g_0$: \[g = e^{2\phi} g_0\]
for some function $\phi: \Sigma \to \bR$. We
also decompose $K$ into a pure trace part and a traceless part,
\[K_{ij} = \frac{\tau}{2} g_{ij} + H_{ij},\] and, following
\cite{ChoquetBruhatMoncrief}, we set
\[\dot{u} \definedas \frac{e^{2u}}{N} \partial_t u.\]
The system \eqref{eqConstraintEquations2} then becomes

\begin{subequations}\label{eqConformalConstraints}
\begin{empheq}[left=\empheqlbrace]{align}
\label{eqVector}
\nabla^i H_{ij} & = - \dot{u} \cdot \partial_j u + \frac{e^{2\phi}}{2} \partial_j \tau,\\
\label{eqLichnerowicz}
\Delta \phi + e^{-2\phi} \left(\frac{1}{2} \dot{u}^2 + \frac{1}{2} \left|H\right|^2\right) & = e^{2\phi} \frac{\tau^2}{4} - \frac{1}{2}\left(1 + \left|\nabla u\right|^2\right),
\end{empheq}
\end{subequations}
where $\nabla$ denotes the Levi-Civita connection of the metric $g_0$, $\Delta$
is the Laplace-Beltrami operator of $g_0$ and from now on, unless stated otherwise,
all norms are taken with respect to the metric $g_0$.

In order to solve Equation \eqref{eqVector}, we split $H$ according to the
York decomposition (see Proposition \ref{propYork} for more details):

\[H = \sigma + LW,\]
where $\sigma$ is a transverse traceless (TT) tensor, i.e. $\tr_{g_0} \sigma \equiv 0$
and $\nabla^i \sigma_{ij} \equiv 0$, and $LW$ denotes the conformal Killing operator
acting on a 1-form $W$:
\[LW_{ij} = \nabla_i W_j + \nabla_j W_i - \nabla^k W_k g_{0ij}.\]

The system \eqref{eqConformalConstraints} finally becomes

\begin{subequations}\label{eqConformalConstraints2}
\begin{empheq}[left=\empheqlbrace]{align}
\label{eqVector2}
-\frac{1}{2} L^* L W & = - \dot{u} \cdot du + \frac{e^{2\phi}}{2} d \tau,\\
\label{eqLichnerowicz2}
\Delta \phi + e^{-2\phi} \left(\frac{1}{2} \dot{u}^2 + \frac{1}{2} \left|\sigma + LW\right|^2\right) & = e^{2\phi} \frac{\tau^2}{4} - \frac{1}{2}\left(1 + \left|\nabla u\right|^2\right),
\end{empheq}
\end{subequations}
where $L^*$ is the formal $L^2$-adjoint of $L$:
\[-\frac{1}{2} L^* L W_j = \nabla^i LW_{ij}.\]

The equations of this system are commonly known as the conformal constraint
equations. Equation \eqref{eqVector2} is called the \emph{vector equation}
and Equation \eqref{eqLichnerowicz2} is named the \emph{Lichnerowicz equation}.

Given $u$, $\dot{u}$, $\tau$ and $\sigma$ we are going to construct solutions
to the system \eqref{eqConformalConstraints2} for the unknowns $\phi$ and $W$
without any smallness assumption on $\tau$. We follow the approach of
\cite{DahlGicquaudHumbert}. The main theorem we prove is the following:

\begin{theorem}\label{thmMain}
Given $\dot{u} \in C^0(\Sigma, \bR)$, $u \in C^1(\Sigma, \bR)$
$\tau \in W^{1, p}(\Sigma, \bR)$ and $\sigma \in W^{1, p}$ a TT-tensor,
where $p > 2$, and assuming that $\tau$ vanishes nowhere on $\Sigma$, then
at least one of the following assertions is true:
\begin{enumerate}
 \item The set of solutions $(\phi, W)$ to the system \eqref{eqConformalConstraints2}
 is non-empty and compact in $W^{2, p}(\Sigma, \bR) \times W^{2, p}(\Sigma, T^*\Sigma)$
 
 \item There exists a non-trivial solution $V \in W^{2, p}(\Sigma, T^* \Sigma)$
of the following limit equation
 \begin{equation}\label{eqLimit}
  -\frac{1}{2} L^* L W = \alpha \frac{\sqrt{2}}{2} \left|LW\right| \frac{d\tau}{|\tau|}
 \end{equation}
 for some $\alpha \in [0, 1]$.
\end{enumerate}
\end{theorem}

\begin{remark}
Since the surface $\Sigma$ is of genus $G\geq2$, there is no conformal Killing
vector fields on $\Sigma$. Therefore $LW \equiv 0$ imply $W \equiv 0$. In particular,
there cannot be any non-zero solution to \eqref{eqLimit} with $\alpha = 0$,
since in this case we would have
\[
0 = \int_\Sigma \left\< W, -\frac{1}{2} L^* L W\right\> d\mu^{g_0} = -\frac{1}{2} \int_\Sigma \left|LW\right|^2 d\mu^{g_0}, 
\]
which immediately implies that $W$ is a conformal Killing vector field.
\end{remark}

The proof of this theorem is the subject of Section \ref{secProof}.

\begin{corollary}\label{corNearCMC}
 Assume that the mean curvature $\tau$ is such that
 \[\left\|\frac{d\tau}{\tau}\right\|_{L^\infty(\Sigma, T^* \Sigma)} < 1\]
 then there exists a solution to the conformal constraint equations
 \eqref{eqConformalConstraints}.
\end{corollary}

See Section \ref{secEx} for the proof of this corollary.

\section{Proof of Theorem \ref{thmMain}}\label{secProof}

Before tackling the full system of equations in Subsection \ref{secCoupled},
we first study the properties of each equation individually, in Subsection
\ref{secVector} for the vector equation and in Subsection
\ref{secLichnerowicz} for the Lichnerowicz equation.

\subsection{The vector equation}\label{secVector}

The main result about Equation \eqref{eqVector} is the following:

\begin{proposition}\label{propVector}
 Given a 1-form $Y \in L^p(\Sigma, T^*\Sigma)$, there exists a unique
 $W \in W^{2, p}(\Sigma, T^*\Sigma)$ such that \[-\frac{1}{2} L^* L W = Y.\]
 Moreover, $W$ satisfies
\[\left\|W\right\|_{W^{2, p}(\Sigma, T^*\Sigma)} \lesssim \left\|Y\right\|_{L^p(\Sigma, T^*\Sigma)}.\]
\end{proposition}

\begin{proof}
We can write
\begin{align}
 -\frac{1}{2} L^* L W_j
 & = \nabla^i\left(\nabla_i W_j + \nabla_j W_i - \nabla^k W_k g_{0ij}\right)\nonumber\\
 & = \Delta W_j + \nabla^i \nabla_j W_i - \nabla_j \nabla^i W_i\nonumber\\
 & = \Delta W_j + \ricdd{i}{j} W^i\nonumber\\
-\frac{1}{2} L^* L W_j
 & = \Delta W_j - \frac{1}{2} W_j,\label{eqBochner}
\end{align}
where we used the fact that in dimension 2, $\ric = \frac{\scal}{2} g_{0ij}$.
This Bochner formula will be useful in Section \ref{secEx}.

On $W^{1, 2}(\Sigma, T^*\Sigma)$, we introduce the following bilinear
form
\[a(V, W) \definedas \int_\Sigma \left\<LV, LW\right\> d\mu^{g_0}.\]
We have
\begin{align*}
a(V, W)
 & = \int_\Sigma \left\<V, L^*LW\right\> d\mu^{g_0}\\
 & = -2 \int_\Sigma \left\<V, \Delta W - \frac{1}{2} W\right\> d\mu^{g_0}\\
 & = \int_\Sigma \left(2\left\<\nabla V, \nabla W\right\> + \left\<V, W\right\>\right) d\mu^{g_0}\\
\end{align*}

It follows immediately that the bilinear form $a$ satisfies the assumptions
of the Lax-Milgram theorem: it is continuous and coercive. So given
$Y \in L^p(\Sigma, T^*\Sigma) \subset \left(W^{1, 2}(\Sigma, T^*\Sigma)\right)^*$
there exists a unique
$W \in W^{1, 2}(\Sigma, T^*\Sigma)$ such that $-\frac{1}{2} L^* L W = Y$.
It follows from elliptic regularity that $W \in W^{2, p}(\Sigma, T^*\Sigma)$
and that
$\left\|W\right\|_{W^{2, p}(\Sigma, T^*\Sigma)} \lesssim \left\|Y\right\|_{L^p(\Sigma, T^*\Sigma)}$.
\end{proof}

In particular, we get the following result:

\begin{proposition}\label{propYork}
Given a symmetric traceless tensor $H \in W^{1, p}$, there exist a unique
TT-tensor $\sigma$ and a unique 1-form $W$ such that \[H = \sigma + LW.\]
\end{proposition}

\begin{proof}
From the previous proposition, there exists a unique solution $W \in W^{2, p}$ of
\[-\frac{1}{2} L^* L W = \divg_{g_0} H.\]
Setting $\sigma = H - LW$, we have
\[
\divg_{g_0} \sigma = \divg_{g_0} H - \divg_{g_0} LW
= \divg_{g_0} H + \frac{1}{2} L^* L W = 0.
\]
Therefore, $\sigma$ is a TT-tensor.
\end{proof}

\subsection{The Lichnerowicz equation}\label{secLichnerowicz}

The aim of this section is to prove the following proposition :

\begin{prop}\label{propLichnerowicz}
Let $\dot{u}$, $u$ and $\tau$ be given as in Theorem \ref{thmMain}. For any
given symmetric traceless 2-tensor $H \in L^\infty$, there exists a unique
positive function $\phi \in W^{2, p}(\Sigma, \bR)$ solving Equation
\eqref{eqLichnerowicz}. Further $\phi$ depends continuously on
$H \in C^0$ and is bounded from below by a positive constant $\phi_0$
which is independent of $H$.
\end{prop}

Before proving the proposition, we need to recall a general lemma on
semilinear elliptic equations. This is a simple version of the so-called
sub and super-solution method we took from \cite[Chapter 14]{Taylor3}.

\begin{lemma}\label{lmSubSuperSolutions}
 Given an open interval $I \subset \bR$, we consider the following equation
 for $\phi$ on $\Sigma$:
 \begin{equation}\label{eqSemilinear}
  \Delta \phi = f(x, \phi, \lambda),
 \end{equation}
where $\lambda \in \Lambda$ is a parameter belonging to $\Lambda$, an open subset of
Banach space, and $f$ is a function belonging to $C^0(\Sigma, \bR) \otimes C^1(I \times \Lambda, \bR)$,
i.e. $f$ decomposes as a finite sum \[f = \sum_i a_i(x) f_i(\phi, \lambda),\] where
$a_i \in C^0(\Sigma, \bR)$ and $f_i \in C^1(I \times \Lambda, \bR)$. We assume further that
 \begin{itemize}
  \item $\pdiff{f}{\phi} > 0$,
  \item there exist constants $a_0, a_1 \in I$ (that may depend continuously on
  $\lambda$), $a_0 \leq a_1$, such that, for all $x \in \Sigma$, $f(x, a_0, \lambda) < 0$
  and $f(x, a_1, \lambda) > 0$.
 \end{itemize}
 Then the equation \eqref{eqSemilinear} admits a unique solution
 $\phi \in W^{2,p}(\Sigma, \bR)$, $2 < p <\infty$, for all $\lambda \in \Lambda$.
 Further, $\phi$ depends continuously on $\lambda$.
\end{lemma}

\begin{proof}
 We first prove the existence of a solution for all $\lambda \in \Lambda$.
 We denote by $\Omega$ the closed subset of $C^0(M, \bR)$ defined
 by \[\Omega = \{\phi \in C^0(M, \bR), a_0 \leq \phi \leq a_1\}.\]
 We choose a constant $A = A(\lambda) > 0$ such that
 \[A > \sup_{(x, \phi) \in \Sigma \times [a_0, a_1]} \pdiff{f}{\phi}(x, \phi, \lambda)\]
 and define a map $F: \Omega \to C^0(M, \bR)$ as follows. Given
 $\phi_0 \in \Omega$, we define $F(\phi_0) \definedas \phi_1$, where
 $\phi_1 \in W^{2, p}(\Sigma, \bR)$ is the (unique) solution to the following
 linear equation:
 \[-\Delta \phi_1 + A \phi_1 = A \phi_0 - f(x, \phi_0, \lambda).\]
 We argue that $\phi_1 \in \Omega$ as follows. We have
 \begin{align*}
 -\Delta \phi_1 + A \phi_1
   & = A \phi_0(x) - f(x, \phi_0, \lambda)\\
   & = \int_{a_0}^{\phi_0(x)} \underbrace{\left(A - \pdiff{f}{\phi}(x, \phi, \lambda)\right)}_{> 0} d\phi + A a_0 - f(x, a_0, \lambda)\\
   & \geq A a_0 - f(x, a_0, \lambda)\\
   & \geq A a_0;\\
 -\Delta \left(\phi_1 -a_0\right) + A \left(\phi_1(x) - a_0\right)
   & \geq 0.
 \end{align*}
 We set $(\phi_1-a_0)_- \definedas \min\{0, \phi_1-a_0\}$. Multiplying the
 previous inequality by $(\phi_1-a_0)_-$ and integrating over $\Sigma$, we get
 \begin{align*}
 \int_\Sigma \left[-(\phi_1-a_0)_-\Delta \left(\phi_1 -a_0\right) + A \left(\phi_1(x) - a_0\right)_-^2\right] d\mu^g
   & \leq 0,\\
\int_\Sigma \left[\left|\nabla(\phi_1-a_0)_-\right|^2 + A \left(\phi_1(x) - a_0\right)_-^2\right] d\mu^g
   & \leq 0,
 \end{align*}
 from which we immediately conclude that $\left(\phi_1(x) - a_0\right)_- \equiv 0$, i.e.
 that $\phi_1 \geq a_0$. A similar argument proves that $\phi_1 \leq a_1$.
 Hence $F$ maps $\Omega$ into itself.
 
 We note that for fixed $\lambda$, $F$ maps $\Omega$ into a bounded
 subset of $W^{2, p}(\Sigma, \bR)$. This comes from the fact that $\Sigma \times [a_0, a_1]$
 is a compact set over which $f(\cdot, \cdot, \lambda)$ is continuous so
 $f(x, \phi, \lambda)$ is bounded independently of $\phi \in \Omega$ and
 $x \in \Sigma$. Hence, by elliptic regularity
 \begin{align*}
  \left\|F(\phi)\right\|_{W^{2, p}(\Sigma, \bR)}
    & \lesssim \left\|f(x, \phi, \lambda)\right\|_{L^\infty(\Sigma, \bR)}\\
    & \lesssim 1.
 \end{align*}

 Denoting by $\Omega'$ the closure of the convex hull of $F(\Omega)$, it
 follows from the Rellich theorem that $\Omega'$ is a compact convex subset
 of $C^0(\Sigma, \bR)$. By the Schauder fixed point theorem, $F$ admits
 a fixed point $\phi$. $\phi$ then satisfies
 \begin{align*}
 - \Delta \phi + A \phi & = A \phi - f(x, \phi, \lambda)\\
 \Leftrightarrow \Delta \phi & = f(x, \phi, \lambda).
 \end{align*}
 Hence $\phi$ is a solution to \eqref{eqSemilinear} and by
 elliptic regularity, $\phi \in W^{2, p}(\Sigma, \bR)$.
 
 We next prove that the solution to \eqref{eqSemilinear} is unique
 given $\lambda \in \Lambda$. It follows then that $a_0 \leq \phi \leq a_1$.
 Assume given $\phi_1, \phi_2$ two solutions to \eqref{eqSemilinear}. We have
 \begin{align*}
  0 & = - \Delta (\phi_2 - \phi_1) + f(x, \phi_2, \lambda) - f(x, \phi_1, \lambda)\\
    & = - \Delta (\phi_2 - \phi_1) + (\phi_2-\phi_1) \underbrace{\int_0^1 \pdiff{f}{\phi}(x, \phi_1 + y (\phi_2-\phi_1)) dy}_{> 0},
 \end{align*}
 from which we immediately conclude that $\phi_1 \equiv \phi_2$.
 
 We follow a similar strategy to prove that $\phi$ depends continuously
 on $\lambda$. We fix an arbitrary $\lambda_0 \in \Lambda$. There exists
 $\alpha > 0$ such that
 \[\pdiff{f}{\phi}(x, \phi, \lambda_0) \geq \alpha\]
 for all $(x, \phi) \in \Sigma \times [a_0(\lambda_0), a_1(\lambda_0)]$.
 There exist an $\eta_0 > 0$ and $a'_0, a'_1 \in I$ such that $B_{\eta_0}(\lambda_0) \subset \Lambda$,
 $a'_0 \leq a_0(\lambda)$, $a'_1 \geq a_1(\lambda)$ for all $\lambda \in B_{\eta_0}(\lambda_0)$
 and
 \[
  \pdiff{f}{\phi}(x, \phi, \lambda) > \frac{\alpha}{2}
 \]
 on $\Sigma \times [a'_0, a'_1] \times B_{\eta_0}(\lambda_0)$.
 We denote by $\phi_0$ the solution to \eqref{eqSemilinear} with $\lambda = \lambda_0$.
 
 For any $\epsilon > 0$, there exists $\eta > 0$, $\eta < \eta_0$ such that
 \[\left|f(x, \phi_0, \lambda_1) - f(x, \phi_0, \lambda_0)\right| < \frac{\epsilon\alpha}{2}\]
 for all $x \in \Sigma$ and all $\lambda \in B_\eta(\lambda_0)$.
 We denote by $\phi_1$ the solution to \eqref{eqSemilinear} with $\lambda = \lambda_1$
 for an arbitrary $\lambda_1 \in B_\eta(\lambda_0)$:
 
 \[
 \left\lbrace
 \begin{aligned}
  - \Delta \phi_0 + f(x, \phi_0, \lambda_0) & = 0\\
  - \Delta \phi_1 + f(x, \phi_1, \lambda_1) & = 0
 \end{aligned}
 \right.
 \]
 
 Subtracting both equations, we get
 \begin{align}
 0 & = - \Delta (\phi_1 - \phi_0) + f(x, \phi_1, \lambda_1) - f(x, \phi_0, \lambda_0)\nonumber\\
   & = - \Delta (\phi_1 - \phi_0) + f(x, \phi_1, \lambda_1) - f(x, \phi_0, \lambda_1) + f(x, \phi_0, \lambda_1) - f(x, \phi_0, \lambda_0\nonumber)\\
 0 & = - \Delta (\phi_1 - \phi_0) + \int_0^1 \pdiff{f}{\phi}(x, \phi_0 + y (\phi_1-\phi_0), \lambda_1) dy (\phi_1 - \phi_0) + f(x, \phi_0, \lambda_1) - f(x, \phi_0, \lambda_0).
 \label{eqDifference}
 \end{align}
 
 From our assumptions, we have
 \[
  \int_0^1 \pdiff{f}{\phi}(x, \phi_0 + y (\phi_1-\phi_0), \lambda_1) dy > \frac{\alpha}{2}.
 \]
 
 Multiplying Equation \eqref{eqDifference} by
 $(\phi_1-\phi_0-\epsilon)_+ \definedas \max\{0, \phi_1-\phi_0-\epsilon\} \geq 0$, and integrating
 over $\Sigma$, we get
 \begin{align*}
  & \int_\Sigma \left(f(x, \phi_0, \lambda_0) - f(x, \phi_0, \lambda_1)\right) (\phi_1-\phi_0-\epsilon)_+ d\mu^{g_0}\\
  & \qquad = \int_\Sigma \Big[\left\<\nabla (\phi_1-\phi_0-\epsilon)_+, \nabla (\phi_1-\phi_0-\epsilon)_+ \right\>\\
  & \qquad \qquad + \left.\int_0^1 \pdiff{f}{\phi}(x, \phi_0 + y (\phi_1-\phi_0), \lambda_1) dy (\phi_1 - \phi_0)(\phi_1-\phi_0-\epsilon)_+\right] d\mu^{g_0},\\
  & \int_\Sigma \frac{\epsilon\alpha}{2} (\phi_1-\phi_0-\epsilon)_+ d\mu^{g_0}\\
  & \qquad \geq \int_\Sigma \left[\left|\nabla (\phi_1-\phi_0-\epsilon)_+\right|^2 + \frac{\alpha}{2} (\phi_1 - \phi_0)(\phi_1-\phi_0-\epsilon)_+\right] d\mu^{g_0}\\
0 & \geq \int_\Sigma \left[\left|\nabla (\phi_1-\phi_0-\epsilon)_+\right|^2 + \frac{\alpha}{2} \left((\phi_1-\phi_0-\epsilon)_+\right)^2\right] d\mu^{g_0}
 \end{align*}
 Hence $\phi_1-\phi_0 \leq \epsilon$. Similarly,  $\phi_1-\phi_0 \geq -\epsilon$.
 This proves that the function $\Psi$ mapping $\lambda$ to $\phi$ solving
 \eqref{eqSemilinear} is continuous from $\Lambda$ to $C^0(\Sigma, I)$.
 It then follows at once from elliptic regularity that $\Psi$ is continuous
 as a mapping from $\Lambda$ to $W^{2, p}(\Sigma, \bR)$.
\end{proof}

We refer the reader to \cite[Section 6]{MaxwellRoughCompact} for much
stronger versions of the sub and super-solution method. We can now give the
proof of Proposition \ref{propLichnerowicz}:

\begin{proof}[Proof of Proposition \ref{propLichnerowicz}]
The Lichnerowicz equation \eqref{eqLichnerowicz} can be rewritten in the
form \eqref{eqSemilinear}:
\[
\Delta \phi = \underbrace{- e^{-2\phi} \left(\frac{1}{2} \dot{u}^2 + \frac{1}{2} \left|H\right|^2\right) + e^{2\phi} \frac{\tau^2}{4} - \frac{1}{2}\left(1 + \left|\nabla u\right|^2\right)}_{:= f(x, \phi)}.
\]
Since $\tau^2$ is bounded away from zero, the assumption $\pdiff{f}{\phi} > 0$
is readily checked. Choosing $a_0 \definedas - \max \ln|\tau|$, we have
\[e^{2 a_0} \frac{\tau^2}{4} \leq \frac{1}{4}.\] So
\[f(x, a_0) \leq e^{2a_0} \frac{\tau^2}{4} - \frac{1}{2}\left(1 + \left|\nabla u\right|^2\right) \leq \frac{1}{4}-\frac{1}{2}\leq-\frac{1}{4}.\]
Since $f$ is increasing with $\phi$, we immediately get that if $\phi < a_0$,
then $f(x, \phi) < 0$. Let now $a_1 \geq 0$ be such that
\[
e^{2 a_1} \frac{\min \tau^2}{4}
 > \frac{1}{2} \left(1 + \left\|\nabla u\right\|_{L^\infty}^2\right)
 + \frac{1}{2} \left\|\dot{u}\right\|_{L^\infty}^2 + \frac{1}{2} \left\|H\right\|_{L^\infty}^2.
\]
Using the fact that we choose $a_1 \geq 0$, it is a simple matter to check
that \[f(x, a_1) > 0\] and hence if $\phi > a_1$, $f(x, \phi) > 0$.

As a consequence, the Lichnerowicz equation satisfies the assumptions of
Lemma \ref{lmSubSuperSolutions}. This completes the proof of Proposition
\ref{propLichnerowicz}.
\end{proof}

\subsection{The coupled system}\label{secCoupled}

Following \cite{Nguyen}, we use Schaefer's fixed point theorem to study
the coupled system (see \cite[Chapter 11]{GilbargTrudinger}):

\begin{theorem}\label{thmSchaefer}
 Let $X$ be a Banach space and $\Phi: X \to X$ a continuous compact mapping.
 Assume that  the set \[F \definedas \{x \in X, \exists \rho \in [0, 1], x = \rho \Phi(x)\}\]
 is bounded.  Then $\Phi$ has a fixed point: \[\exists x \in X, x = \Phi(x),\]
 and the set of fixed points is compact.
\end{theorem}

We choose $X = C^0(\Sigma, \bR)$ as a Banach space and
construct the mapping $\Phi$ as follows:

Given $v \in X$,
\begin{itemize}
 \item From Proposition \ref{propVector} there exists a unique
 $W \definedas W(v) \in W^{2, p}$ solving
 \begin{equation}\label{eqVectorModified}
 -\frac{1}{2} L^* L W = - \dot{u} \cdot du + \frac{v^2}{2} d \tau,
 \end{equation}
 which is Equation \ref{eqVector2} with $e^{\phi} = v$. Further
 $W$ depends continuously on $v \in C^0$ for the $W^{2, p}$-norm.
 \item $W \in W^{2, p}$ can then be continuously mapped to
 $H \definedas \sigma + L W \in W^{1, p}$
 \item and, in turn, $H$ can be compactly embedded into $C^0$.
 \item Proposition \ref{propLichnerowicz} yields a unique
 $\phi \in W^{2, p}$ solving the Lichnerowicz equation \eqref{eqLichnerowicz}
 with the $H$ we previously found.
\end{itemize}

Setting $\Phi(v) \definedas e^{\phi} \in C^0(\Sigma, \bR)$, we loop the loop
providing a continuous compact map $\Phi: X \to X$. Thus, we are almost under
the assumptions of Theorem \ref{thmSchaefer}. All we need to check is that
the set $F$ is bounded. This is the content of the next proposition:

\begin{proposition}\label{propLimit}
 Assume that the set
 \[F \definedas \{v \in L^{\infty}(\Sigma, \bR), \exists \rho \in [0, 1], v = \rho \Phi(v)\}\]
 is unbounded. Then there exists a constant $\rho_0 \in [0, 1]$ and a
 non-zero $W \in W^{2, p}$ such that
 \[- \frac{1}{2} L^* L W = \frac{\sqrt{2}}{2} \rho_0 \left|L W\right| \frac{d \tau}{|\tau|}.\]
\end{proposition}

\begin{proof}
Assuming that $F$ is unbounded, we can find sequences $(\rho_i)_{i \geq 0}$ and
$(v_i)_{i \geq 0}$ such that $0 \leq \rho_i \leq 1$, $v_i = \rho_i \Phi(v_i)$
and $\left\|v_i\right\|_{L^\infty} \to \infty$. Setting
$\phi_i = \log(\Phi(v_i))$ (i.e. $v_i = \rho_i e^{\phi_i}$), and defining
$W_i$ as the solution to \eqref{eqVectorModified} with $v \equiv v_i$, we
get the following equations:

\begin{subequations}\label{eqConformalConstraints3}
\begin{empheq}[left=\empheqlbrace]{align}
\label{eqVector3}
-\frac{1}{2} L^* L W_i & = - \dot{u} \cdot du + \rho_i^2 \frac{e^{2\phi_i}}{2} d \tau,\\
\label{eqLichnerowicz3}
\Delta \phi_i + e^{-2\phi_i} \left(\frac{1}{2} \dot{u}^2 + \frac{1}{2} \left|\sigma + LW_i\right|^2\right) & = e^{2\phi_i} \frac{\tau^2}{4} - \frac{1}{2}\left(1 + \left|\nabla u\right|^2\right),
\end{empheq}
\end{subequations}

Following \cite{DahlGicquaudHumbert, GicquaudSakovich, Nguyen}, we set
$\gamma_i \definedas \left\|e^{\phi_i}\right\|_{L^\infty}$
and we introduce the following rescaled objects:

\[
\psi_i \definedas \phi_i - \log(\gamma_i),\;
\Wtil_i \definedas \frac{1}{\gamma_i^2} W_i.
\]

Note that since we assumed that $\left\|v_i\right\|_{L^\infty} =
\rho_i \gamma_i \to \infty$, with $0 \leq \rho_i \leq 1$, we also
have that $\gamma_i \to \infty$. We will assume without loss of generality
that $\gamma_i \geq 1$. The following equations for $\psi_i$ and $\Wtil_i$
follow from the definition:

\begin{subequations}\label{eqRescaled1}
\begin{empheq}[left=\empheqlbrace]{align}
\label{eqRescaledVector1}
-\frac{1}{2} L^* L \Wtil_i & = - \frac{1}{\gamma_i^2} \dot{u} \cdot du + \rho_i^2 \frac{e^{2\psi_i}}{2} d \tau,\\
\label{eqRescaledLichnerowicz1}
\frac{1}{\gamma_i^2} \Delta \psi_i + e^{-2\psi_i} \left(\frac{1}{2 \gamma_i^4} \dot{u}^2 + \frac{1}{2} \left|\frac{\sigma}{\gamma_i^2} + L\Wtil_i\right|^2\right)
  & = e^{2\psi_i} \frac{\tau^2}{4} - \frac{1}{2 \gamma_i^2}\left(1 + \left|\nabla u\right|^2\right),
\end{empheq}
\end{subequations}

It follows from the definition of $\gamma_i$ that
$\left\|e^{\psi_i}\right\|_{L^\infty} = \left\|\frac{1}{\gamma_i} e^{\phi_i}\right\|_{L^\infty} = 1$.
Hence, from Proposition \ref{propVector} applied to \eqref{eqRescaledVector1},
we have

\begin{align*}
 \left\|\Wtil_i\right\|_{W^{2, p}}
   & \lesssim \left\|- \frac{1}{\gamma_i^2} \dot{u} \cdot du + \rho_i^2 \frac{e^{2\psi_i}}{2} d \tau\right\|_{L^p}\\
   & \lesssim \frac{1}{\gamma_i^2} \left\|\dot{u} \cdot du\right\|_{L^p} +  \left\|d\tau\right\|_{L^p}\\
   & \lesssim 1.
\end{align*}

Consequently, $\Wtil_i$ is bounded in $W^{2, p}$. Since the embedding
$W^{2, p} \into C^1$ is compact, we can assume, up to extraction, that
$\Wtil_i$ converges to some $\Wtil_\infty \in W^{2, p}$ for the $C^1$-norm.
We can also assume that $\rho_i \to \rho_\infty \in [0, 1]$. All we need to
do is to prove that $e^{2\psi_i}$ converges in $L^\infty$ to
$f_\infty \definedas \sqrt{2} \frac{\left|L \Wtil_\infty\right|}{|\tau|}$.

Indeed, passing to the limit in Equation \eqref{eqRescaledVector1}, we get
that $\Wtil_\infty$ satisfies

\begin{align}
-\frac{1}{2} L^* L \Wtil_\infty
 & = \rho_\infty^2 \frac{f_\infty}{2} d\tau\nonumber\\
 & = \frac{\sqrt{2}}{2} \rho_\infty^2 \left|L \Wtil_\infty\right|\frac{d\tau}{|\tau|}.
   \label{eqLimit0}
\end{align}

Hence, $\Wtil_\infty$ satisfies the limit equation with
$\alpha = \rho_\infty^2$. Since $e^{2 \psi_i}$ has $L^\infty$-norm 1 and
converges in $L^\infty$ to $f_\infty$, we have
$\left\|f_\infty\right\|_{L^\infty} = 1$. In particular,
$L \Wtil_\infty \not\equiv 0$ which proves that $\Wtil_\infty \not\equiv 0$.

To prove convergence of $e^{2\psi_i}$ to $f_\infty$, we show that for any
$\epsilon > 0$, there exists an $i_0$ such that
\[\left|e^{2 \psi_i} - f_\infty\right| \leq \epsilon\]
for all $i \geq i_0$. We do it in two steps:

\begin{itemize}
 \item We first show the upper bound
\[e^{2 \psi_i} \leq f_\infty + \epsilon\]
by selecting a smooth function $f_+$ such that
\[f_\infty + \frac{\epsilon}{2} \leq f_+ \leq f_\infty + \epsilon\]
and proving that for $i_0$ large enough,
$\psi_+ \definedas \frac{1}{2} \log(f_+)$ is a super-solution to
\eqref{eqRescaledLichnerowicz1}:

\begin{equation}\label{eqRescaledLichnerowicz2}
\frac{1}{\gamma_i^2} \Delta \psi_+ + e^{-2\psi_+} \left(\frac{1}{2 \gamma_i^4} \dot{u}^2 + \frac{1}{2} \left|\frac{\sigma}{\gamma_i^2} + L\Wtil_i\right|^2\right)
  \leq e^{2\psi_+} \frac{\tau^2}{4} - \frac{1}{2 \gamma_i^2}\left(1 + \left|\nabla u\right|^2\right).
\end{equation}

Since $f_\infty \geq 0$, $f_+ \geq \frac{\epsilon}{2}$ so $\psi_+$ is
a smooth function. In particular, $\left|\Delta \psi_+\right|$ is
bounded. Moreover, since $\Wtil_i \to \Wtil_\infty$ in $C^1$ and
$\gamma_i \to \infty$, we have
\[\left|\frac{\sigma}{\gamma_i^2} + L\Wtil_i\right|^2 \to \left|L\Wtil_\infty\right|^2\]
as $i$ tends to infinity. So the condition \eqref{eqRescaledLichnerowicz2}
can be rephrased as
\[o(1) + \frac{1}{2} \left|L \Wtil_\infty\right|^2 - \frac{\tau^2}{4} f_+^2 \leq 0,\]
where $o(1)$ denotes a sequence of functions tending uniformly to 0
when $i \to \infty$. We have
\[f_+^2 \geq \left(f_\infty + \frac{\epsilon}{2}\right)^2 \geq f_\infty^2 + \frac{\epsilon^2}{4}.\]
This yields, for $i$ big enough,
\[
o(1) + \frac{1}{2} \left|L \Wtil_\infty\right|^2 - \frac{\tau^2}{4} f_+^2
\leq o(1) + \frac{\tau^2}{4} f_\infty^2 -  \frac{\tau^2}{4} \left(f_\infty^2 + \frac{\epsilon^2}{4}\right)
\leq o(1) - \frac{\tau_0^2 \epsilon^2}{4} \leq 0,
\]
where $\tau_0^2 \definedas \inf_\Sigma \tau^2$ is positive by assumption.
Therefore $\psi_+$ is a super-solution to Equation \eqref{eqRescaledLichnerowicz1}
and we obtain, for $i$ big enough
\begin{align*}
\frac{1}{\gamma_i^2} \Delta(\psi_+-\psi_i)
 & \leq -\left(e^{-2\psi_+}-e^{-2\psi_i}\right)\left(\frac{\dot{u}^2}{2 \gamma_i^4} +\frac{1}{2}\left|L \Wtil_i+\frac{\sigma}{\gamma_i^2}\right|^2\right)
    + \frac{\tau^2}{4} \left(e^{2\psi_+}-e^{2 \psi_i}\right)\\
 & \leq \frac{\tau^2}{2}e^{2\psi_i} (\psi_+-\psi_i)\int_0^1 e^{2\lambda (\psi_+-\psi_i)}d\lambda\\
 &  \qquad + \left(\frac{\dot{u}^2}{2 \gamma_i^4} +\frac{1}{2}\left|L \Wtil_i+\frac{\sigma}{\gamma_i^2}\right|^2\right) e^{-2\psi_i}(\psi_+-\psi_i)\int_0^1e^{-2\lambda(\psi_+-\psi_i)}d\lambda\\
 & \leq \underbrace{\left[\frac{\tau^2}{2}e^{2\psi_i} \int_0^1 e^{2\lambda (\psi_+-\psi_i)}d\lambda
         + \left(\frac{\dot{u}^2}{2 \gamma_i^4} +\frac{1}{2}\left|L \Wtil_i+\frac{\sigma}{\gamma_i^2}\right|^2\right) e^{-2\psi_i}\int_0^1e^{-2\lambda(\psi_+-\psi_i)}d\lambda \right]}_{> 0} (\psi_+-\psi_i).
\end{align*}
The maximum principle implies that ${\psi_i} \leq \psi_+$, for $i$
big enough, so
\[
e^{2 \psi_i}\leq f_\infty +\epsilon. 
\]
\item Second, we show the lower bound
\[e^{2 \psi_i} \geq f_\infty - \epsilon\]
We have to be more careful than for the super-solution, since $f_\infty$ can
vanish somewhere. Let $f_{-}$ be a smooth function such that
\[
\sqrt{\max(f_\infty^2-\epsilon,0)}\leq f_{-}\leq \sqrt{\max (f_\infty^2-\frac{\epsilon}{2},0)}.
\]

We will work on the open domain $\cA$ defined by 
\[\cA= \{x \in \Sigma, f_{-}(x)>0\}.\]

On $\cA$, we can define $\psi_- = \frac{1}{2}\ln(f_-)$. We want to show that the
following inequality is satisfied on $\cA$:

\begin{equation}\label{subsol}
\frac{1}{\gamma_i^2} \Delta \psi_- + e^{-2 \psi_-}\left(\frac{1}{2\gamma_i^4}\dot{u}^2 +\frac{1}{2}\left|\frac{\sigma}{\gamma_i^2} + L \Wtil_i\right|^2\right)
 \geq e^{2 \psi_-} \frac{\tau^2}{4} - \frac{1}{2\gamma_i^2}(1+|\nabla u|^2).
\end{equation}

Since $e^{2 \psi_-}>0$ on $\cA$, that is equivalent to showing
\[
\frac{1}{\gamma_i^2} e^{2 \psi_-} \left(\Delta \psi_- +\frac{1}{2}(1+|\nabla u|^2)\right) + \left(\frac{1}{2 \gamma_i^4}\dot{u}^2 +\frac{1}{2}\left|\frac{\sigma}{\gamma_i^2}+L \Wtil_i\right|^2\right)-e^{4 \psi_-}\frac{\tau^2}{4}\geq 0.
\] 

We calculate
\[e^{2\psi_-} \Delta \psi_-= \frac{1}{2} \left[ \Delta f_- - \frac{\left|\nabla f_-\right|^2}{f_-}\right].\]

We can assume that $\partial \cA$ is the disjoint union of smooth curves and denote by $r$
the signed distance function to $\partial \cA$
which is positive where $f_\infty \geq \epsilon$. We choose $f_-$ such that $f_- \equiv 0$ whenever
$r \leq 0$ and $f_- \equiv \epsilon e^{-1/r}$ if $r > 0$ is sufficiently small for some positive $\epsilon$.
For such a choice of $f_-$, $e^{2\psi_-}\Delta \psi_-$ is bounded on $\cA$.

Therefore, as for the upper bound, the condition \eqref{subsol} can be written
\[o(1) +\frac{1}{2}|LW_\infty|^2-e^{4\psi_-}\frac{\tau^2}{4}\geq 0.\]
On $\cA$ we have $e^{4\psi_-}\leq f^2_-\leq f_\infty^2-\frac{\epsilon}{2}$. This yields for $i$ big enough 
\[o(1) +\frac{1}{2}|LW_\infty|^2  -e^{4\psi_-}\frac{\tau^2}{4} \geq 
o(1)+\frac{\tau^2}{4}f_\infty^2-\frac{\tau^2}{4}\left(f_\infty^2-\frac{\epsilon}{2}\right) \geq
o(1)+\frac{\tau^2}{4}\frac{\epsilon}{2}\geq 0.\]
Since $\psi_{-}(x)-\psi_i(x)\rightarrow-\infty$ when $x\rightarrow \partial \cA$, $\psi_{-}(x)-\psi_i(x)$
attains its maximum on $\cA$. Therefore, since $\psi_{-}$ is a subsolution, we can apply the maximum principle on
$\cA$, to deduce that $\psi_- \leq \psi_i$. 
This yields
\[\max(f_\infty^2-\epsilon,0) \leq e^{4 \psi_i}.\]
\end{itemize}
This concludes the proof of the convergence in $L^\infty$ of $e^{2 \psi_i}$ towards $f_\infty$.
\end{proof}

\section{Proof of Corollary \ref{corNearCMC}}\label{secEx}

To prove Corollary \ref{corNearCMC}, all we need to do is to prove that
the limit equation \eqref{eqLimit} admits no non-zero solution under the
assumption
\[\left\|\frac{d\tau}{\tau}\right\|_{L^\infty(\Sigma, T^* \Sigma)} < 1.\]

We take the scalar product of the limit equation with $W$ and integrate
over $\Sigma$. From the Bochner formula \eqref{eqBochner}, we get:

\begin{align*}
-\frac{1}{2} \int_\Sigma \left|LW\right|^2 d\mu^{g_0}
 & = \alpha \frac{\sqrt{2}}{2} \int_\Sigma \left|LW\right| \left\<W, \frac{d\tau}{|\tau|}\right\> d\mu^{g_0}\\
\int_\Sigma \left|\nabla W\right|^2 d\mu^{g_0} + \frac{1}{2} \int_\Sigma \left|W\right|^2 d\mu^{g_0}
 & \leq \alpha \sqrt{2} \int_\Sigma \left|\nabla W\right| \left|\frac{d\tau}{\tau}\right| |W| d\mu^{g_0}\\
 & \leq \alpha \int_\Sigma \left|\nabla W\right|^2 d\mu^{g_0} + \frac{\alpha}{2}\int_\Sigma \left|\frac{d\tau}{\tau}\right|^2 |W|^2 d\mu^{g_0}\\
\frac{1}{2} \int_\Sigma \left|W\right|^2 d\mu^{g_0}
 & \leq \frac{\alpha}{2} \left\|\frac{d\tau}{\tau}\right\|_{L^\infty}^2 \int_\Sigma |W|^2 d\mu^{g_0},
\end{align*}
where we used the well-known inequality $ab \leq \frac{a^2}{2} + \frac{b^2}{2}$
with $a = \sqrt{2} \left|\nabla W\right|$ and $b = \left|\frac{d\tau}{\tau}\right| |W|$.
The last inequality immediately yields that $W \equiv 0$ since we assumed
that $\left\|\frac{d\tau}{\tau}\right\|_{L^\infty}^2 < 1$ and $\alpha \in [0, 1]$.

\providecommand{\bysame}{\leavevmode\hbox to3em{\hrulefill}\thinspace}
\providecommand{\MR}{\relax\ifhmode\unskip\space\fi MR }
\providecommand{\MRhref}[2]{%
  \href{http://www.ams.org/mathscinet-getitem?mr=#1}{#2}
}
\providecommand{\href}[2]{#2}


\end{document}